\documentclass[12pt]{article}
\usepackage{fullpage,graphicx,psfrag,amsmath,amsfonts,verbatim}
\usepackage[small,bf]{caption}
\usepackage{algorithmicx,algpseudocode}
\usepackage{url,algorithm,amsthm}

\newcommand{\Wmin}{W^\mathrm{min}}
\newcommand{\BEAS}{\begin{eqnarray*}}
\newcommand{\EEAS}{\end{eqnarray*}}
\newcommand{\BEQ}{\begin{equation}}
\newcommand{\EEQ}{\end{equation}}
\newcommand{\BIT}{\begin{itemize}}
\newcommand{\EIT}{\end{itemize}}

\newcommand{\eg}{{\it e.g.}}
\newcommand{\ie}{{\it i.e.}}
\newcommand{\ones}{\mathbf 1}

\newcommand{\reals}{{\mbox{\bf R}}}
\newcommand{\symm}{{\mbox{\bf S}}}

\newcommand{\Expect}{\mathbf{E}}
\newcommand{\prob}{\mathbf{Prob}}
\newcommand{\var}{\mathop{\bf var}}

{\begin{quote}}{\end{quote}}

\makeatletter
\long\def\@makecaption#1#2{
   \vskip 9pt 
   \begin{small}
   \setbox\@tempboxa\hbox{{\bf #1:} #2}
   \ifdim \wd\@tempboxa > 5.5in
        \begin{center}
        \begin{minipage}[t]{5.5in}
        \addtolength{\baselineskip}{-0.95pt}
        {\bf #1:} #2 \par
        \addtolength{\baselineskip}{0.95pt}
        \end{minipage}
        \end{center}
   \else 
	\hbox to\hsize{\hfil\box\@tempboxa\hfil}  
   \fi
   \end{small}\par
}
\makeatother

\newcounter{oursection}

\newcounter{lecture}

\usepackage{subcaption}
\usepackage{lscape} % for horizontal pages with results
\usepackage{booktabs}
\usepackage{placeins}

\usepackage{hyperref} %%for clickable references!!
\usepackage{amsthm}
\newtheorem{lem}{Lemma}
\usepackage[toc,page]{appendix}
\graphicspath{{../Graphics/}}

\title{Risk-Constrained Kelly Gambling}
\author{Enzo Busseti \and Ernest K. Ryu \and Stephen Boyd}

\date{\today}

\begin{document}
\maketitle
\begin{abstract}
We consider the classic Kelly gambling problem
with general distribution of outcomes, and an
additional risk constraint that limits the probability
of a drawdown of wealth to a given undesirable level.
We develop a bound on the drawdown probability;
using this bound instead of the original risk constraint
yields a convex optimization problem that guarantees 
the drawdown risk constraint holds.  
Numerical experiments show that
our bound on drawdown probability is reasonably close
to the actual drawdown risk, as computed by Monte Carlo
simulation.
Our method is parametrized by a single parameter that 
has a natural interpretation as a risk-aversion
parameter, allowing us to systematically trade off asymptotic 
growth rate and drawdown risk.
Simulations show that this method yields bets that 
out perform fractional-Kelly
bets for the same drawdown risk level or growth rate.
Finally, we show that a natural quadratic approximation of our
convex problem is closely connected to the classical
mean-variance Markowitz portfolio selection problem.
\end{abstract}

\section{Introduction}

In 1956 John Kelly proposed a systematic way to allocate a total wealth
across a number of bets so as to maximize the long term growth rate
when the gamble is repeated \cite{kelly1956new,maclean2011kelly}. 
Similar results were later derived in the finance literature,
under the name of \emph{growth-optimum portfolio};
see, \eg, \cite[Ch.\ 6]{merton1990continuous}).
It is well known that with Kelly optimal bets
there is a risk of the wealth dropping substantially 
from its original value before increasing, \ie, a \emph{drawdown}.
Several ad hoc methods can be used to limit this drawdown risk,
at the cost of decreased growth rate.
The best known method is fractional Kelly betting,
in which only a fraction of the 
Kelly optimal bets are made \cite{davis2012fractional}.
The same method has been proposed in the finance 
literature \cite{browne2000risk}.
Another ad hoc method is
Markowitz's mean-variance portfolio optimization
\cite{markowitz1952portfolio}, which trades off two objectives
that are related to, but not the same as, long term growth rate
and drawdown risk.

In this paper we directly address drawdown risk
and show how to find bets that trade off drawdown risk and 
growth rates.
We introduce the risk-constrained Kelly gambling problem,
in which the long term wealth growth rate is maximized
with an additional constraint that limits the probability
of a drawdown to a specified level.
This idealized problem captures 
what we want, but seems very difficult to solve.
We then introduce a convex optimization problem that is 
a \emph{restriction} of this problem; that is, its feasible
set is smaller than that of the risk-constrained problem.
This problem is tractable, using modern convex optimization
methods.  

Our method can be used in two ways.  First, it can be used
to find a conservative set of bets that are guaranteed to
satisfy a given drawdown risk constraint.  Alternatively, its
single parameter can be interpreted as a risk-aversion parameter
that controls the trade-off between growth rate and drawdown risk,
analogous to Markowitz mean-variance portfolio optimization 
\cite{markowitz1952portfolio},
which trades off mean return and (variance) risk.
Indeed, we show that a natural quadratic approximation
of our convex problem can be closely connected
to Markowitz mean-variance portfolio optimization.

In \S\ref{s-kelly} we review the Kelly gambling problem, and describe methods
for computing the optimal bets using convex optimization.
In simple cases, such as when there are two possible 
outcomes, the Kelly optimal bets are well known.  In others, for example when
the returns come from infinite distributions, the methods do not seem to be well known.
In \S\ref{s-drawdown}, we define the drawdown risk, and in \S\ref{s-risk-bound}, we derive 
a bound on the drawdown risk.  In \S\ref{s-rck}, we use this bound to form 
the risk-constrained Kelly gambling problem, which is a tractable convex optimization
problem.
In \S\ref{s-quadratic}, we derive a quadratic approximation of the risk-constrained
Kelly gambling problem, and relate it to classical Markowitz portfolio optimization.
Finally, in \S\ref{s-examples} we give some numerical examples to illustrate the 
methods.

\section{Kelly gambling}\label{s-kelly}
% XXX the fact that the allocation vector is fixed in time is
% not completely trivial. it follows from our objective (maximizing growth of log wealth).
% we could say or derive it XXX
In Kelly gambling, we place a fixed fraction of our total 
wealth (assumed positive) on $n$ bets.  
We denote the fractions 
as $b \in \reals^n$, so $b \geq 0$ and $\ones^Tb=1$, where
$\ones$ is the vector with all components $1$.
The $n$ bets have a random nonnegative payoff or return, 
denoted $r\in \reals^n_+$,
so the wealth after the bet changes by the (random) factor $r^Tb$.
We will assume that all bets do not have infinite return in expectation, \ie, that
$\Expect r_i<\infty$ for $i=1,\dots,n$.
We will also assume that the bet $n$ has a certain return of one,
\ie, $r_n =1$ almost surely.  This means that $b_n$ represents the
fraction of our wealth that we do not wager, or hold 
as cash.  The bet vector $b=e_n$ corresponds to not betting at all.
We refer to the bets $1,\ldots, n-1$ as the risky bets.

We mention some special cases of this general Kelly gambling setup.
\BIT
\item \emph{Two outcomes.} We have $n=2$, and $r$ takes on only
two values: $(P,1)$, with probability $\pi$,
and $(0,1)$, with probability $1-\pi$.
The first outcome corresponds to winning the bet
($\pi$ is the probability of winning) and $P>1$ is the payoff.
\item \emph{Mutually exclusive outcomes.} 
There are $n-1$ mutually exclusive outcomes, with return vectors:
$r = (P_ke_k,1)$ with probability $\pi_k > 0$, for $k=1, \ldots, n-1$,
where $P_k>1$ is the payoff for outcome $k$, and $e_k$ is the 
unit vector with $k$th entry one and all others $0$.
Here we bet on which of outcomes $1,\ldots, n-1$ will be the 
winner (\eg, the winner of a horse race).
\item \emph{General finite outcomes.}  The return vector takes
on $K$ values
$r_1, \ldots, r_K$, with probabilities $\pi_1, \ldots, \pi_K$.
This case allows for more complex bets, for example in horse racing
show, place, exacta, perfecta, and so on.
\item \emph{General returns.} 
The return $r$ comes from an arbitrary infinite distribution
(with $r_n=1$ almost surely).
If the returns are log-normal, the gamble is
a simple model of investing 
(long only) in $n-1$ assets with log-normal returns;
the $n$th asset is risk free (cash).
More generally, we can have $n-1$ arbitrary derivatives (\eg, options)
with payoffs that depend on an underlying random variable.
\EIT

\subsection{Wealth growth}

The gamble is repeated at times (epochs) $t=1,2,\ldots$,
with IID (independent and identically distributed) returns.
Starting with initial wealth $w_1=1$, 
the wealth at time $t$ is given by
\[
w_t = (r_1^Tb) \cdots (r_{t-1}^T b),
\]
where $r_t$ denotes the realized return at time $t$
(and not the $t$th entry of the vector).
The wealth sequence $\{w_t\}$ is a stochastic process
that depends on the choice of the bet vector $b$, as well as
the distribution of return vector $r$.
Our goal is to choose $b$ so that, roughly speaking,
the wealth becomes large.

Note that $w_t \geq 0$, since $r \geq 0$ and $b \geq 0$.
The event $w_t=0$ is called \emph{ruin}, and can only happen
if $r^Tb=0$ has positive probability.
The methods for choosing $b$ that we discuss below
all preclude ruin, so we will assume it does not occur,
\ie, $w_t>0$ for all $t$, almost surely.
Note that if $b_n>0$, ruin cannot occur since
$r^T b \geq b_n$ almost surely.

With $v_t = \log w_t$ denoting the logarithm of the wealth, we have
\[
v_t = \log (r_1^Tb) + \cdots + \log (r_{t-1}^Tb).
\]
Thus $\{v_t\}$ is a random walk, with increment 
distribution given by the distribution of $\log (r^T b)$.
The drift of the random walk is $\Expect \log (r^Tb)$;
we have $\Expect v_t = (t-1) \Expect \log (r^Tb)$, and
$\var v_t  = (t-1)\var \log(r^Tb)$.
%(We assume that both of these exist.)
The quantity $\Expect \log(r^Tb)$ can be interpreted as the 
average growth rate of the wealth; it is the drift in the 
random walk $\{v_t\}$.
The (expected) growth rate $\Expect \log(r^Tb)$ is a 
function of the bet vector $b$.

\subsection{Kelly gambling}
In Kelly gambling, we choose $b$ to maximize $\Expect \log (r^Tb)$,
the growth rate of wealth. 
This leads to the optimization problem
\BEQ\label{e-kelly}
\begin{array}{ll}
\text{maximize} & \Expect \log (r^Tb) \\
\text{subject to} & \ones^T b = 1, \quad b \geq 0,
\end{array}
\EEQ
with variable $b$.
We call a solution $b^\star$ of this problem a set of \emph{Kelly optimal} bets.
The Kelly gambling problem 
is always feasible, since $b=e_n$ (which corresponds
to not placing any bets) is feasible.
This choice achieves objective value zero, so the optimal value 
of the Kely gambling problem is always nonnegative.
The Kelly gambling problem \eqref{e-kelly} is a convex 
optimization problem,
since the objective is concave and the constraints are convex.

Kelly optimal bets maximize growth rate of the wealth.
If $\tilde b$ is a bet vector that is not Kelly optimal, with
associated wealth sequence $\tilde w_t$,
and $b^\star$ is Kelly optimal, with associated wealth sequence
$w^\star_t$,
then  $w^\star_t/\tilde w_t \to \infty$ with probability one as $t \to \infty$.
(This follows immediately since the random walk $\log w^\star_t - 
\log \tilde w_t$ has positive drift \cite[\S XII.2]{feller1971}.
Also see \cite[\S16]{cover2012elements} for a general discussion of
Kelly gambling.)

We note that
the bet vector $b=e_n$ is Kelly optimal
if and only if
$\Expect r_i \leq 1$ for $i=1, \ldots, n-1$.
Thus we should not bet at all
if all the bets are losers in expectation; 
conversely, if just one bet is a winner in expectation,
the optimal bet is not the trivial one $e_n$, and
the optimal growth rate is positive.
We show this in the appendix.
%(See Lemma~\ref{lem-loser-bet} of the appendix.)

\subsection{Computing Kelly optimal bets}
Here we describe methods to compute Kelly optimal bets, 
\ie, to solve the Kelly 
optimization problem~(\ref{e-kelly}). 
It can be solved analytically or semi-analytically for simple cases;
the general finite outcomes case can be handled by standard 
convex optimization tools, and the general case can be handled
via stochastic optimization.

\paragraph{Two outcomes.}
For a simple bet with two outcomes, with win probability $\pi$ and payoff $P$,
we obtain the optimal bet with simple minimization of a univariate function.
We have 
\[
b^\star = \left(\frac{\pi P - 1}{P-1},\frac{P-\pi P}{P-1}\right),
\]
provided $\pi P>1$; if $\pi P\leq 1$, the optimal bet is $b^\star=(0,1)$.
Thus we should bet 
a fraction $(\pi P - 1)/(P-1)$ of our wealth each time,
if this quantity is positive.

\paragraph{General finite outcomes.}
When the return distribution is finite the Kelly gambling problem
reduces to 
\BEQ\label{e-kelly-finite}
\begin{array}{ll}
\text{maximize} & \sum_{i=1}^K \pi_i \log (r_i^Tb) \\
\text{subject to} & \ones^T b = 1, \quad b \geq 0,
\end{array}
\EEQ
which is readily solved using convex optimization methods
\cite{boyd2004convex}.
Convex optimization software systems like CVX \cite{cvx}, 
CVXPY \cite{cvxpy} and Convex.jl \cite{convexjl},
based on DCP (Disciplined Convex Programming) \cite{grant2006dcp}, or
others like YALMIP \cite{lofberg2004yalmip},
can handle such problems directly. In our numerical simulation we
use CVXPY with the open source solver ECOS \cite{domahidi2013ecos}, 
recently extended to handle exponential
cone constraints \cite{akle2015algorithms}.

\paragraph{General returns.}
We can solve the Kelly gambling problem~(\ref{e-kelly}) even in
the most general case, when $r$ takes on an infinite number of values,
provided we can generate samples from the distribution of $r$.
In this case we can use a projected stochastic gradient method with averaging
\cite{robbins1951, nemirovski1983, polyak1987introduction,kushner2003,bubeck2015}.

As a technical assumption we assume here that the 
Kelly optimal bet $b^\star$ satisfies $(b^\star)_n>0$,
\ie, the optimal bet involves holding some cash.
We assume we know $\varepsilon>0$ that satisfies $\varepsilon< (b^\star)_n$
(which implies that $r^Tb > \varepsilon$ a.s.)
and define $\Delta_\varepsilon=\{b \mid \ones^Tb=1,~b \geq 0,~b_n\ge \varepsilon\}$.
Then the gradient of the objective is given by
\[
\nabla_b \Expect \log(r^Tb) =
\Expect\nabla_b  \log(r^Tb) = \Expect \frac{1}{r^Tb} r
\]
for any $b\in \Delta_\varepsilon$.
So if $r^{(k)}$ is an IID sample from the distribution,
\BEQ\label{e-sg}
\frac{1}{{r^{(k)}}^Tb} r^{(k)}
\EEQ
is an unbiased estimate of the gradient,
\ie, a stochastic gradient, of the objective  at $b$.
(Another unbiased estimate of the gradient can be obtained by averaging 
the expression \eqref{e-sg} over multiple return samples.)

The (projected) stochastic gradient method with averaging
computes the iterates 
\[
\bar b^{(k+1)} = \Pi \left( \bar b^{(k)} + 
\frac{t_k}{({r^{(k)}}^T \bar b^{(k)})}r^{(k)} \right), \quad k=1,2, \ldots,
\]
where the starting point $\bar b^{(1)}$ is any vector in $\Delta_\varepsilon$,
$r^{(k)}$ are IID samples from the distribution of $r$,
and $\Pi$ is (Euclidean) projection onto $\Delta_\varepsilon$
(which is readily computed; see
Lemma~\ref{lem-proj}).
%, \eg, 
%\cite[\S6.2.5]{parikh2014proximal}, \cite{duchi2008efficient} check again XXX). 
The step sizes $t_k>0$ must satisfy
\[
\qquad t_k\rightarrow 0,
\qquad
\sum^\infty_{k=1}t_k=\infty.
\]
(For example, $t_k=C/\sqrt{k}$ with any $C>0$ satisfies this condition.)
Then the (weighted) running average
\[
{b}^{(k)} = \frac{\sum_{i=1}^k t_i \bar b^{(i)}}{\sum_{i=1}^k t_i}
\]
converges to Kelly optimal.
The stochastic gradient method can be slow to converge,
but it always works;
that is $\Expect\log (r^T {b}^{(k)})$ converges to the optimal growth rate.

In practice, one does not know a priori how small $\varepsilon$ should be.
One way around this is to choose a small $\varepsilon$,
and then check that $({b}^{(k)})_n> \varepsilon$ holds for large $k$,
in which case we know our guess of $\varepsilon$ was valid.
A more important practical variation on the algorithm is \emph{batching},
where we replace the unbiased estimate of the gradient with the average 
over some number of samples.  This does not affect the 
theoretical convergence of the algorithm, but can improve 
convergence in practice.

\section{Drawdown}\label{s-drawdown}
We define the \emph{minimum wealth} as
the infimum of the wealth trajectory over time,
\[
\Wmin = \inf_{t=1,2, \ldots } w_t.
\] 
This is a random variable, with distribution that depends on $b$.
With $b=e_n$, we have $w_t=1$ for all $t$, so $\Wmin=1$.
With $b$ for which $\Expect \log (r^Tb)>0$ (which we assume),
$\Wmin$ takes values in $(0,1]$. 
Small $\Wmin$ corresponds to a case where the initial wealth drops
to a small value before eventually increasing.

The \emph{drawdown} is defined as $1-\Wmin$.
A drawdown of $0.3$ means that
our wealth dropped $30\%$ from its initial value (one), before
increasing (which it eventually must do, 
since $v_t \to \infty$ with probability one).
Several other definitions of drawdown are used in the literature.
A large drawdown means that $\Wmin$ is small, \ie, our wealth drops to
a small value before growing.

The \emph{drawdown risk} is defined as
$\prob (\Wmin < \alpha)$,
where $\alpha \in (0,1)$ is a given target (undesired) minimum wealth.
This risk depends on the bet vector $b$ in a very complicated way.
There is in general no formula for the risk in terms of 
$b$, but we can always (approximately) compute the drawdown risk for 
a given $b$ using Monte Carlo simulation.
As an example, a drawdown risk of $0.1$ for $\alpha = 0.7$ means 
the probability of experiencing more than $30\%$ drawdown is only $10\%$.
The smaller the drawdown risk (with any target), the better.

\subsection{Fractional Kelly gambling}
It is well known that Kelly optimal bets can lead to 
substantial drawdown risk.
One ad hoc method for handling this is to
compute a Kelly optimal bet $b^\star$, and then use the so-called
fractional Kelly \cite{davis2012fractional} bet given by
\BEQ\label{e-fractional-kelly}
b = fb^\star + (1-f)e_n,
\EEQ
where $f\in (0,1)$ is the fraction. The fractional
Kelly bet scales down the (risky) bets by $f$.
Fractional Kelly bets have smaller drawdowns than Kelly bets, at 
the cost of reduced growth rate.
We will see that trading off growth rate and drawdown risk
can be more directly (and better) handled.

\subsection{Kelly gambling with drawdown risk}
We can add a drawdown risk constraint to the Kelly gambling 
problem~(\ref{e-kelly}), to obtain the problem
%\emph{risk-constrained Kelly gambling problem}
\BEQ\label{e-kelly-risk}
\begin{array}{ll}
\text{maximize} & \Expect \log (r^Tb) \\
\text{subject to} & \ones^T b = 1, \quad b \geq 0,\\
& \prob (\Wmin < \alpha) <\beta,
\end{array}
\EEQ
with variable $b$, where $\alpha,\beta \in (0,1)$ are
given parameters.
The last constraint limits the probability of a drop in wealth to
value $\alpha$ to be no more than $\beta$.
For example, we might take 
$\alpha = 0.7$ and $\beta = 0.1$, meaning that we require the 
probability of a drawdown of more than $30\%$ to be less than 10\%.
(This imposes a constraint on the bet vector $b$.)

Unfortunately the problem~(\ref{e-kelly-risk}) is,
as far as we know, a difficult optimization problem in general.
In the next section we will develop a bound on the drawdown risk
that results in a tractable convex constraint on $b$.
We will see in numerical simulations that the bound is 
generally quite good.

\section{Drawdown risk bound}\label{s-risk-bound}
In this section we derive a condition that bounds the drawdown risk.
Consider any $\lambda>0$ and bet $b$.
For any $\alpha\in(0,1)$ and $\beta\in(0,1)$ that satisfies
$\lambda = \log \beta/\log \alpha$ we have
\BEQ
\Expect (r^Tb)^{-\lambda} \leq 1 ~\Longrightarrow ~ \prob(\Wmin < \alpha) <
\beta. 
\label{e-bound-2}
\EEQ
In other words, if our bet satisfies $\Expect (r^Tb)^{-\lambda} \leq 1$,
then its drawdown risk $\prob(\Wmin < \alpha)$ less than $\beta$.

To see this, consider the stopping time
\[
\tau = \inf\{t\ge 1\,|\,w_t< \alpha\},
\]
and note $\tau<\infty$ if and only if $\Wmin< \alpha$.
From Lemma~\ref{lem-rw} of the appendix, we get
\begin{align*}
1\ge 
\Expect
\left[
\exp(-\lambda \log w_\tau - \tau \log \Expect(r^Tb)^{-\lambda} )
\mid
\tau<\infty
\right]
\prob(\Wmin< \alpha).
\end{align*}
Since
$-\tau \log \Expect(r^Tb)^{-\lambda} \ge 0$
when $\tau<\infty$,
we have
\[
1\ge
\Expect
\left[
\exp(-\lambda \log w_\tau )
\mid
\tau<\infty
\right]
\prob(\Wmin< \alpha).
\]
Since $w_\tau< \alpha$
when $\tau<\infty$, we have
\[
1>
\exp(-\lambda \log \alpha)
\prob(\Wmin< \alpha).
\]
So we have
\[
\prob(\Wmin< \alpha)< \alpha^\lambda=\beta.
\]

\section{Risk-constrained Kelly gambling}
\label{s-rck}
Replacing the drawdown risk constraint in the problem
\eqref{e-kelly-risk} with the lefthand side of
\eqref{e-bound-2}, with $\lambda = \log \beta/\log \alpha$,
yields the {risk-constrained Kelly gambling problem} (RCK)
\BEQ\label{e-kelly-risk-restr}
\begin{array}{ll}
\text{maximize} & \Expect \log (r^T b),\\
 \text{subject to} & \ones^T b = 1, \quad b \geq 0,\\
 & \Expect(r^Tb)^{-\lambda} \leq 1,
\end{array}
\EEQ
with variable $b$.
We refer to a solution of this problem as an RCK bet.
The RCK problem is a \emph{restriction} of 
problem~\eqref{e-kelly-risk}, since it has a smaller feasible set:
any $b$ that is feasible for RCK must satisfy the 
drawdown risk constraint $\prob(\Wmin < \alpha) < \beta$.
In the limit when either $\beta \to 1$
or $\alpha \to 0$ we get $\lambda \to 0$.  For $\lambda=0$,
the second constraint is always satisfied, 
and the RCK problem \eqref{e-kelly-risk-restr} 
reduces to the (unconstrained) Kelly gambling problem \eqref{e-kelly}.

Let us now show that the RCK problem \eqref{e-kelly-risk-restr} is convex.
The objective is concave and the constraints $\ones^Tb=1$, $b \geq 0$ are
convex.
To see that the function $\Expect (r^Tb)^{-\lambda}$ is convex in $b$,
we note that for $r^Tb>0$, 
$(r^Tb)^{-\lambda}$ is a convex function of $b$; so the expectation
$\Expect (r^Tb)^{-\lambda}$ is a convex function of $b$ (see 
\cite[\S 3.2]{boyd2004convex}).
We mention that the last constraint can also be written as 
$\log \Expect (r^Tb)^{-\lambda} \leq 0$, where the lefthand side is
a convex function of $b$.

The RCK problem \eqref{e-kelly-risk-restr} is 
always feasible, since $b=e_n$ is feasible.
Just as in the Kelly gambling problem,
the bet vector $b=e_n$ is optimal for RCK \eqref{e-kelly-risk-restr}
if and only if
$\Expect r_i \leq 1$ for $i=1, \ldots, n-1$.
In other words we should not bet at all
if all the bets are losers in expectation; conversely, 
if just one bet is a winner in expectation, the solution of the 
RCK problem will have positive
growth rate (and of course respect the drawdown risk constraint).
We show this in the appendix.

%Under our assumption that the Kelly optimal bet is not $e_n$,
%the optimal value is positive, 
%and it is strictly feasible (which implies that Slater's condition holds).
%XXX our assumptions so far don't guarantee that no element of $b^\star$ is zero.XXX
%To see this, 
%\[
%\phi'(f) = \frac{-\lambda}{\Expect (r^Tb)^{-\lambda}}
%\Expect (r^Tb)^{-\lambda -1} r^T(b^\star - e_n),
%\]
%so
%\[
%\phi'(0) = -\lambda \left( \Expect (r^T b^\star)  - 1\right) <0,
%\]
%since $\Expect (r^Tb^\star) \geq \exp \Expect \log(r^Tb^\star) > 1$ by
%Jensen's inequality.

\subsection{Risk aversion parameter}\label{s-risk-aversion-par}
The RCK problem~\eqref{e-kelly-risk-restr}
 depends on the parameters $\alpha$ and $\beta$ only through
$\lambda = \log \beta/\log \alpha$.
This means that, for fixed $\lambda$, our one constraint
$ \Expect(r^Tb)^{-\lambda} \leq 1$
actually gives us a family of drawdown constraints
that must be satisfied:
\BEQ\label{e-cdf-risk-bound}
\prob(\Wmin < \alpha) < \alpha ^\lambda
\EEQ
holds for all $\alpha \in (0,1)$.
For example $\alpha = 0.7$ and $\beta = 0.1$ gives $\lambda = 6.46$;
thus, our constraint also implies that the probability of a drop
in wealth to $\alpha = 0.5$ (\ie, we lose half our initial wealth) 
is no more than $(0.5)^{6.46} = 0.011$.
Another interpretation of \eqref{e-cdf-risk-bound} 
is that our risk constraint actually bounds the entire
CDF (cumulative distribution function) of $\Wmin$: it stays 
below the function $\alpha \mapsto \alpha^\lambda$.

The RCK problem~\eqref{e-kelly-risk-restr}
 can be used in two (related) ways.
First, we can start from the original drawdown specifications,
given by $\alpha$ and $\beta$, and then solve the problem
using $\lambda = \log \beta/\log \alpha$.
In this case we are guaranteed that the resulting bet satisfies our drawdown constraint. 
An alternate use model is to consider $\lambda$ as a \emph{risk-aversion
parameter}; we vary it to trade off growth rate and drawdown risk.
This is very similar to traditional Markowitz portofolio optimization
\cite{markowitz1952portfolio} \cite[\S4.4.1]{boyd2004convex},
where a risk aversion parameter is used to trade off
risk (measured by portfolio return variance) and return (expected portfolio
return).
We will see another close connection between our method and 
Markowitz mean-variance portfolio optimization in \S\ref{s-quadratic}.

\subsection{Light and heavy risk aversion regimes}
In this section we give provide an interpretation of the drawdown risk 
constraint
\BEQ\label{e-kelly-constraint}
\Expect(r^Tb)^{-\lambda} \leq 1
\EEQ
in the light and heavy risk aversion regimes, which correspond to small
and large values of $\lambda$, respectively.

\paragraph{Heavy risk aversion.}
In the heavy risk aversion regime,
\ie, in the limit $\lambda\rightarrow \infty$,
the constraint \eqref{e-kelly-constraint}
reduces to $r^Tb \geq 1$ almost surely.
In other words, 
problem~\eqref{e-kelly-risk-restr} only considers risk-free bets
in this regime.

\paragraph{Light risk aversion.}
Next consider the light risk aversion regime,
\ie, in the limit $\lambda\rightarrow 0$.
Note that constraint \eqref{e-kelly-constraint} is equivalent to
\[
\frac{1}{\lambda}\log \Expect\exp(-\lambda\log(r^Tb)) \leq 0.
\]
As $\lambda\rightarrow 0$ we have
\begin{align*}
\frac{1}{\lambda}\log \Expect\exp(-\lambda\log(r^Tb))&=
\frac{1}{\lambda}\log \Expect\left[
1-\lambda \log(r^Tb)+\frac{\lambda^2}{2}(\log(r^Tb))^2+O(\lambda^3)\right]\\
&=
-\Expect\log(r^Tb)+\frac{\lambda}{2}\Expect(\log(r^Tb))^2
-\frac{\lambda}{2}(\Expect\log(r^Tb))^2
+O(\lambda^2)\\
&=
-\Expect\log(r^Tb)+\frac{\lambda}{2}\var \log(r^Tb) +O(\lambda^2).
\end{align*}
(In the context of stochastic control,
$(1/\lambda)\log \Expect_X\exp(-\lambda X)$
is known as the exponential disutility or loss
and $\lambda$ as the risk-sensitivity parameter.
This asymptotic expansion is well-known
see \eg\ \cite{whittle1981,whittle1990}.)
So constraint \eqref{e-kelly-constraint}
%which can be expressed as $\log \Expect (r^Tb)^{-\lambda} \leq 0$,
reduces to
\[
\frac{\lambda}{2}\var\log (r^Tb)
\le
\Expect\log (r^Tb)
\]
in the limit $\lambda \to 0$.
Thus the (restricted) drawdown risk contraint \eqref{e-kelly-constraint}
limits the ratio of variance to mean growth in this regime.

%XXX 
%Figuring out the asymptotic expansion about $\lambda=\infty$ is equivalent
%to finding the asymptotic expansion of
%\[
%\lim_{\lambda\rightarrow\infty}(\Expect X^\lambda)^{1/\lambda}=\|X\|_\infty
%\]
%about $\lambda=\infty$.
%The answer is complicated
%and its form depends heavily on $r$'s distribution.
%I don't think we should include it.XXX

\subsection{Computing RCK bets}
\paragraph{Two outcomes.}
For the two outcome case we can easily solve the 
problem \eqref{e-kelly-risk-restr}, almost analytically.
The problem is
\BEQ\label{e-kelly-risk-two-outcomes}
\begin{array}{ll}
\text{maximize} & \pi \log (b_1 P + (1-b_1)) + (1-\pi)(1-b_1),\\
 \text{subject to} &0 \leq  b_1 \leq 1,\\
 &  \pi (b_1 P + (1-b_1))^{-\log \beta/\log \alpha} + (1-\pi)(1-b_1)^{-\log \beta/\log \alpha} \leq 1.
\end{array}
\EEQ
If the solution of the unconstrained problem,
\[
\left(\frac{\pi P - 1}{P-1},\frac{P-\pi P}{P-1}\right),
\]
satisfies the risk constraint, then it is the solution.
Otherwise we reduce $b_1$ to find the value for which
\[
\pi (b_1 P + (1-b_1))^{-\lambda} + (1-\pi)(1-b_1)^{-\log \lambda} = 1.
\]
(This can be done by bisection since the lefthand side is monotone in
$b_1$.)
In this case the RCK bet 
%of \eqref{e-kelly-risk-restr} 
is a fractional Kelly bet \eqref{e-fractional-kelly}, for some $f<1$.

\paragraph{Finite outcomes case.}
For the finite outcomes case we can 
restate the RCK problem \eqref{e-kelly-risk-restr}
in a convenient and tractable form.
We first take the log of the last constraint and get
\[
\log \sum_{i=1}^K \pi_i (r_i^Tb)^{-\lambda} \leq 0,
\]
we then write it as 
\[
\log \left( \sum_{i=1}^K \exp (\log \pi_i - \lambda \log (r_i^Tb)) \right) \leq 0.
\]
To see that this constraint is convex we note that the log-sum-exp function
is convex and increasing, and its arguments are all convex functions of $b$
(since $\log (r_i^Tb)$ is concave), so the lefthand side
function is convex in $b$ \cite[\S 3.2]{boyd2004convex}.
Moreover, convex optimization software systems like CVX, CVXPY, and Convex.jl 
based on DCP (disciplined convex programming) can handle such compositions
directly.
Thus we have the problem
\BEQ
\label{e-problem-finite-outcome}
\begin{array}{ll}
\text{maximize} & \sum_{i=1}^K \pi_i \log (r_i^T b),\\
 \text{subject to} & \ones^T b  = 1, \quad b \geq 0,\\
 & \log \left( \sum_{i=1}^K \exp(\log \pi_i -\lambda \log (r_i^Tb) )\right) \leq 0.
\end{array}
\EEQ
In this form the problem is readily solved;
its CVXPY specification is given in appendix~\ref{s-dcp-spec}.

\paragraph{General returns.}
As with the Kelly gambling problem,
we can solve the RCK 
problem \eqref{e-kelly-risk-restr} using
a stochastic optimization method.
We use a primal-dual stochastic gradient method 
(from \cite{nemirovski1978,nemirovski2009robust})
applied to the Lagrangian
\BEQ\label{e-lagrangian}
L(b, \kappa) = -\Expect \log (r^T b) + \kappa (\Expect (r^T b)^{-\lambda} - 1),
\EEQ
with $b\in \Delta_\varepsilon$ and $\kappa \geq 0$.
(As in the unconstrained Kelly optimization case, we make the technical assumption
that $b_n^\star> \varepsilon>0$.)
In the appendix we show that 
the RCK problem \eqref{e-kelly-risk-restr} 
has an optimal dual variable $\kappa^\star$
for the constraint $\Expect (r^T b)^{-\lambda} \le 1$,
which implies that solving problem \eqref{e-kelly-risk-restr} is
equivalent to finding a saddle point of \eqref{e-lagrangian}.
We also assume we know an upper bound $M$ on the value of the
optimal dual variable $\kappa^\star$.

Our method computes the iterates
\begin{align*}
 \bar b^{(k+1)} =& \Pi \left( \bar b^{(k)} + t_k
\frac{ ({r^{(k)}}^T \bar b^{(k)})^{\lambda}+
\lambda \bar \kappa^{(k)}}{({r^{(k)}}^T \bar b^{(k)})^{\lambda+1}}
r^{(k)}
\right),\\
\bar \kappa^{(k+1)} =&
\left( \bar \kappa^{(k)} + t_k\frac{1-({r^{(k)}}^T \bar b^{(k)})^{\lambda}}{
({r^{(k)}}^T \bar b^{(k)})^{\lambda}}
 \right)_{[0,M]},
\end{align*}
where the starting points $\bar b^{(1)}$ and $\bar \kappa^{(1)}$ are respectively
in $\Delta_\varepsilon$ and $[0,M]$,
$r^{(k)}$ are IID samples from the distribution of $r$,
$\Pi$ is Euclidean projection onto $\Delta_\varepsilon$,
and $\left( a \right)_{[0,M]}$ is projection onto $[0,M]$,
\ie,
\[
\left( a \right)_{[0,M]} = \max \{0, \min \{ M,a \} \}.
\]
The step sizes $t_k>0$ must satisfy
\[
\qquad t_k\rightarrow 0,
\qquad
\sum^\infty_{k=1}t_k=\infty.
\]

We use the (weighted) running averages
\[
{b}^{(k)} = \frac{\sum_{i=1}^k t_i \bar b^{(i)}}{\sum_{i=1}^k t_i},
\qquad
{\kappa}^{(k)} = \frac{\sum_{i=1}^k t_i \bar \kappa^{(i)}}{\sum_{i=1}^k t_i}
\]
as our estimates of the optimal bet and $\kappa^\star$, respectively.

Again, this method can be slow to converge, but it always works;
that is $\Expect\log(r^T{b}^{(k)})$ converges to the optimal value
and $\max\{\Expect(r^T{b}^{(k)}))^{-\lambda}-1, 0\}\rightarrow 0$.
As in the unconstrained Kelly case, we do not know a priori how small
$\varepsilon$ should be, or how large $M$ should be.
We can choose a small $\varepsilon$ and large $M$
and later verify that $(\bar{b}^{(k)})_n> \varepsilon$,
and $\bar{\kappa}^{(k)}<M$;
if this holds, our guesses of $\varepsilon$ and $M$ were valid.
Also as in the unconstrained case, batching can be used to improve
the practical convergence.  In this case, we replace our 
unbiased estimates of the gradients of the two expectations
with an average over some number of them.

Finally, we mention that the optimality conditions can be 
independently checked.
As we show in Lemma~\ref{lem-opt-cond} of the appendix,
a pair $(b^\star,\kappa^\star)$ is a solution of the RCK problem
if and only if it satisfies the following optimality conditions:
\begin{align}\label{e-opt-cond}
&\ones^Tb^\star=1,\quad b^\star\ge 0,\quad\Expect (r^Tb^\star)^{-\lambda}\le 1\nonumber\\
&\kappa^\star\ge 0,\quad
\kappa^\star(\Expect (r^Tb^\star)^{-\lambda}- 1)=0\\
&\Expect \frac{r_i}{r^Tb^\star}+\kappa^\star\lambda \Expect\frac{r_i}{(r^Tb^\star)^{\lambda+1}}
~\left\{
\begin{array}{ll}
\le 1+\kappa^\star\lambda & b_i=0\\
= 1+\kappa^\star\lambda & b_i>0.
\end{array}\right. \nonumber
\end{align}
These conditions can be checked for a computed approximate solution
of RCK, using Monte Carlo simulation to evaluate the expectations.
(The method above guarantees that $\ones^Tb=1$, $b\geq 0$, and
$\kappa \geq 0$, so we only need to check the other 
three conditions.)

\iffalse
\[
\begin{array}{c}
{\partial L(b, \kappa)}/{\partial b} = 0,\\
\kappa (\Expect (r^T b)^{-\lambda} - 1) = 0, \quad \kappa \geq 0, \\
 \ones^T b  = 1, \quad b \geq 0.\\
\end{array}
\]

 we write the partial Lagrangian
 \[
 L(b, \nu) = \Expect \log (r^T b) - \nu (\Expect (r^T b)^{-\lambda} - 1)
 \]
 keeping the constraints $\ones^T b  = 1$, $b \geq 0$. We maximize $L$ over $b$
 \[
 \nabla_b L = \Expect \frac{1}{r^Tb}r + \lambda \nu \Expect \frac{1}{(r^T b)^{\lambda+1}}r
 \]
\fi
\section{Quadratic approximation} \label{s-quadratic}
In this section we form a quadratic approximation of the 
RCK problem~\eqref{e-kelly-risk-restr},
which we call the \emph{quadratic RCK problem} (QRCK),
and derive a close connection to Markowitz portfolio optimization.
We use the notation $\rho=r-\ones$ for the (random) excess 
return, so (with $\ones^Tb=1$) we have $r^Tb-1 = \rho^Tb$.
Assuming $r^Tb\approx 1$, or equivalently $\rho^Tb \approx 0$,
we have the (Taylor) approximations
\begin{eqnarray*}
\log (r^Tb) &=& \rho^T b -\frac{1}{2}(\rho^Tb)^2 + O((\rho^Tb)^3),\\
(r^Tb)^{-\lambda}
% = (\rho^Tb + 1)^{-\lambda} 
&=& 1 - \lambda \rho^Tb + \frac{\lambda(\lambda+1)}{2}
(\rho^Tb )^2 + O((\rho^Tb)^3).
\end{eqnarray*}
Substituting these into the RCK problem 
\eqref{e-kelly-risk-restr} we obtain the QRCK problem
\begin{equation}\label{e-qp-approx}
\begin{array}{ll}
\mbox{maximize}&\Expect \rho^Tb-\frac{1}{2}\Expect(\rho^Tb)^2\\
\mbox{subject to}& \ones^Tb = 1, \quad b \geq 0\\
&-\lambda \Expect \rho^Tb + \frac{\lambda(\lambda+1)}{2} \Expect (\rho^Tb)^2
\leq 0.
\end{array}
\end{equation}
This approximation of the RCK problem
is a convex quadratic program (QP) which is readily solved.
We expect the solution to be a good approximation of the
RCK solution when the basic assumption $r^Tb \approx 1$ holds.

This approximation can be useful when finding 
a solution to 
the RCK problem~\eqref{e-kelly-risk-restr}.
We first 
estimate the first and second moments of $\rho$ via Monte Carlo,
and then solve the QRCK problem~\eqref{e-qp-approx} (with the estimated moments)
to get a solution $ b^\mathrm{qp}$ and
a Langrange multipler $ \kappa^\mathrm{qp}$.
We take these as good approximations for the solution 
of \eqref{e-kelly-risk-restr}, and use them as the starting points for 
the primal-dual stochastic gradient method,
\ie, we set $b^{(1)}=b^\mathrm{qp}$ and $\kappa^{(1)}= \kappa^\mathrm{qp}$.
This gives no theoretical
advantage, since the method converges no matter
what the initial points are; but it can speed up the convergence in 
practice.

We now connect the QRCK problem~\eqref{e-qp-approx}
to classical Markowitz portfolio
selection.  We start by defining $\mu = \Expect \rho$, the mean excess return,
and
\[
S = \Expect \rho\rho^T = \Sigma + (\Expect \rho)(\Expect \rho)^T,
\]
the (raw) second moment of $\rho$ (with $\Sigma$ the covariance of the return).
We say that an allocation vector $b$ is a Markowitz portfolio if it solves
\BEQ
\begin{array}{ll}\label{e-qp-markowitz}
\mbox{maximize}&\mu^Tb-\frac{\gamma}{2}b^T \Sigma b\\
\mbox{subject to}&  \ones^Tb = 1, \quad b \geq 0,
\end{array}
\EEQ
for some value of the (risk-aversion) parameter $\gamma \geq 0$.
A solution to problem \eqref{e-qp-approx} is a Markowitz portfolio,
provided there are no arbitrages. 
By no arbitrage we mean that $\mu^Tb>0$, $\ones^Tb=1$, and $b\ge 0$, implies
$b^T\Sigma b>0$.

Let us show this.
Let $b^\mathrm{qp}$ be the solution to the 
QRCK problem \eqref{e-qp-approx}.
By (strong) Lagrange duality \cite{bertsekas2009}, $b^\mathrm{qp}$ is a solution of
\begin{equation*}\label{e-qp-approx2}
\begin{array}{ll}
\mbox{maximize}&\mu^Tb-\frac{1}{2}b^TSb+\nu (\mu^Tb-\frac{\lambda+1}{2}b^TSb)\\
\mbox{subject to}& \ones^T b =1, \quad b \geq 0
\end{array}
\end{equation*}
for some $\nu\ge 0$, which we get by dualizing only the constraint
$- \lambda\Expect \rho^Tb + \frac{\lambda(\lambda+1)}{2} \Expect (\rho^Tb)^2\le 0$.
We divide the objective by $1+\nu$ and substitute $S=\mu\mu^T+\Sigma$
to get that $b^\mathrm{qp}$ is a solution of
\begin{equation}\label{e-qp-approx3}
\begin{array}{ll}
\mbox{maximize}&\mu^Tb-\frac{\eta}{2}(\mu^Tb)^2-\frac{\eta}{2}b^T \Sigma b\\
\mbox{subject to}& \ones^T b =1, \quad b \geq 0,
\end{array} 
\end{equation}
for some $\eta>0$.
%Note that $f(x)=x-\eta/2x^2$ is increasing on $x<1/\eta$ and decreasing on $x>1/\eta$.
In turn, $b^\mathrm{qp}$ is a solution of
\BEQ\label{e-qp-lin}
\begin{array}{ll}
\mbox{maximize}&
(1-\eta \mu^Tb^\mathrm{qp})\mu^Tb-\frac{\eta}{2}b^T\Sigma b\\
\mbox{subject to}
&\ones^T b =1, \quad b \geq 0,
\end{array}
\EEQ
since the objectives of problem \eqref{e-qp-approx3} and \eqref{e-qp-lin}
have the same gradient at $b^\mathrm{qp}$.
If $\mu^Tb^\mathrm{qp}<1/\eta$ then problem \eqref{e-qp-lin} is equivalent to
problem \eqref{e-qp-markowitz} with $\gamma=\eta/(1-\eta \mu^Tb^\mathrm{qp})$.

Assume for contradiction that $\mu^Tb^\mathrm{qp}\ge 1/\eta$, which implies 
$(b^\mathrm{qp})^T\Sigma b^\mathrm{qp}> 0$ by the no artibtrage assumption.
Then for problem \eqref{e-qp-lin} the bet $e_n$ achieves objective value $0$,
which is better than that of $b^\mathrm{qp}$. As this contradicts the optimality of $b^\mathrm{qp}$,
we have $\mu^Tb^\mathrm{qp}<1/\eta$. So we conclude a solution to problem \eqref{e-qp-approx}
is a solution to problem \eqref{e-qp-markowitz}, \ie, $b^\mathrm{qp}$ is a Markowitz portfolio.

\section{Numerical simulations}\label{s-examples}
In this section we report results for two specific problem instances, 
one with finite outcomes
and one with infinite outcomes, but our numerical
explorations show that these results are typical.

\subsection{Finite outcomes}
We consider a finite outcomes case
with $n=20$ (so there are $19$ risky bets) and $K = 100$ possible outcomes.
The problem data is generated as follows.
The probabilities $\pi_i$, $i = 1, \ldots, K$
are drawn uniformly on $[0,1]$ and then normalized 
so that $\sum_i \pi_i = 1$.
The returns $r_{ij} \in \reals_{++}$ for $i = 1,\ldots, K$ and $j = 1,\ldots,n-1$ 
are drawn from a uniform distribution in $[0.7, 1.3]$. 
Then, 30 randomly selected returns $r_{ij}$ 
are set equal to $0.2$ and other 30 equal to $2$. 
(The returns $r_{in}$ for $i = 1,\ldots, K$ are instead all set to 1.)
The probability that a return vector $r$
contains at least one ``extreme'' return (\ie, equal to $0.2$ or $2$) 
is $1 - (1 - 60/(19\cdot100))^{n-1} \approx 0.45$.

\subsubsection{Comparison of Kelly and RCK bets}
We compare the Kelly optimal bet with
the RCK bets for $\alpha = 0.7$ and $\beta = 0.1$ ($\lambda = 6.456$).
We then obtain the RCK bets for $\lambda = 5.500$,
a value chosen so that we
achieve risk close to the 
specified value $\beta = 0.1$
(as discussed in \S\ref{s-risk-aversion-par}).
For each bet vector we carry out 10000 Monte Carlo 
simulations of $w_t$ for $t=1, \ldots, 100$.
This allows us to estimate (well) the associated risk probabilities.
Table \ref{tab-results-montecarlo} shows the results.
The second column gives the growth rate, the third column gives the
bound on drawdown risk, and the last column gives the
drawdown risk computed by Monte Carlo simulation.
\begin{table}
\begin{center}
\begin{tabular}{l|c|c|c}
Bet & $\Expect \log (r^Tb)$ & $e^{\lambda \log \alpha}$ &
%$\Expect (r^Tb)^{-\lambda}$ bound &
$\prob(\Wmin < \alpha)$\\ \hline
Kelly  & 0.062  & -  & 0.397 \\
RCK, $\lambda =6.456$  & 0.043  & 0.100  & 0.073 \\
RCK, $\lambda =5.500$  & 0.047  & 0.141  & 0.099 \\
%QRCK, $\lambda = 0.000$ & 0.054  & 1.000  & 0.218 \\ 
%QRCK, $\lambda = 6.456$ & 0.027  & 0.100  & 0.025 \\
%QRCK, $\lambda = 2.800$ & 0.044  & 0.368  & 0.100 \\
\end{tabular}
\caption{Comparison of Kelly and RCK bets. 
Expected growth rate
and drawdown risk are computed with Monte Carlo simulations.}
\label{tab-results-montecarlo}
\end{center}
\end{table}

The Kelly optimal bet experiences a drawdown exceeding our 
threshold $\alpha = 0.7$ around 40\% of the time.
For all the RCK bets the drawdown risk (computed by Monte Carlo) 
is less than our bound, but not dramatically so.
(We have observed this over a wide range of problems.)
The RCK bet with $\lambda = 6.456$ is \emph{guaranteed} to have
a drawdown probability not exceeding $10\%$; Monte Carlo simulation
shows that it is (approximately) $7.3\%$.
For the third bet vector in our comparison, we decreased
the risk aversion parameter
until we obtained a bet with (Monte Carlo computed) risk
near the limit $10\%$.

The optimal value of
the (hard) Kelly gambling problem with 
drawdown risk \eqref{e-kelly-constraint} 
must be less than
$0.062$ (the unconstrained optimal growth
rate) and greater than
$0.043$ (since our second bet vector is guaranteed to satisfy 
the risk constraint).
Since our third bet vector has drawdown risk less than $10\%$,
we can further refine this result to state that
the optimal value of
the (hard) Kelly gambling problem with 
drawdown risk \eqref{e-kelly-constraint} 
is between $0.062$ and $0.047$.

%Modulo Monte Carlo error (which is small),
%we can assert that it (probably) lies between $0.047$ 
%(the growth rate achieved by our third bet vector, which probably
%satisfies the drawdown risk limit), and $0.062$.  

%We can assert as a mathematical fact that the optimal value of
%the (hard) Kelly gambling problem with 
%drawdown risk \eqref{e-kelly-constraint} is between
%$0.043$ (since our second bet vector is guaranteed to satisfy 
%the risk constraint) and $0.062$ (the unconstrained optimal growth
%rate).  Modulo Monte Carlo error (which is small),
%we can assert that it (probably) lies between $0.047$ 
%(the growth rate achieved by our third bet vector, which probably
%satisfies the drawdown risk limit), and $0.062$.  
%It is very likely that our third bet is nearly optimal for the 
%(hard) risk-constrained problem \eqref{e-kelly-constraint}.
%For the QRCK bets we have no guarantee about the drawdown risk.
%However, in this instance we observe that the bound holds.
%The value $\lambda = 2.800$ is selected so that 
%the risk is approximately $0.10$. Its growth rate is smaller than the RCK bet with $\lambda = 5.500$. 

Figure \ref{fig-time-domain-trajectories} shows ten
trajectories of $w_t$ in our Monte Carlo simulations for 
the Kelly optimal bet (left) and the RCK bet obtained with $\lambda = 5.5$
(right).
Out of these ten simulations, four of the Kelly trajectories dip
below the threshold $\alpha = 0.7$, and one of the other trajectories
does, which is consistent with the probabilities reported above.
%XXX We should make it so that 4 of the 10 kelly trajectories drop below $\alpha$.XXX

%Note that one of the trajectories for the Kelly optimal bet
%involves a disastrous drawdown, where the wealth drops 
%in the 23rd time period to only $20\%$ of its inital value.
\begin{figure}
\begin{center}
\includegraphics[width=1.\textwidth]{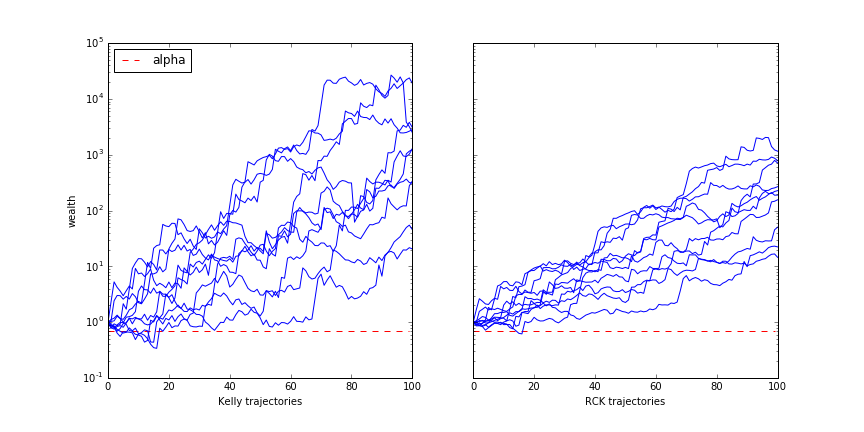}
\end{center}
\caption{Wealth trajectories for the Kelly optimal bet (left)
and the restricted risk-constrained bet with $\lambda=5.5$.
The dashed line shows the wealth threshold $\alpha = 0.7$.}
\label{fig-time-domain-trajectories}
\end{figure}

Figure \ref{fig-empirical-cdf} shows the sample CDF of $\Wmin$
over the 10000 simulated trajectories of $w_t$,
for the Kelly optimal bets
and the RCK bets with $\lambda = 6.46$.
The upper bound $\alpha^\lambda$ is also shown.
We see that the risk bound is not bad, typically around
30\% or so higher than the actual risk. 
We have observed this to be the case across many problem instances.
\begin{figure}
\begin{center}
\includegraphics[width=.7\textwidth]{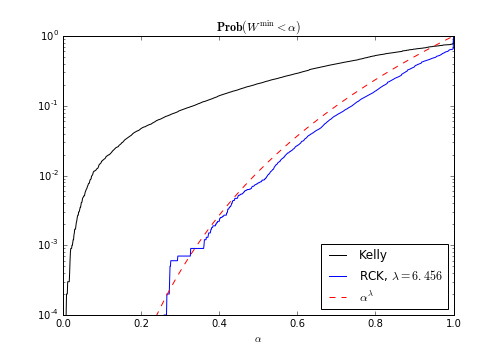}
\end{center}
\caption{Sample CDF of $\Wmin$ 
(\ie, $\prob(\Wmin < \alpha)$) 
for the Kelly optimal bets and RCK bets 
with $\lambda = 6.46$.
The upper bound $\alpha^\lambda$ is also shown.
%XXX We should run more iterations so that the upper bond
%does indeed bound the CDF dawdown probability
%like Figure \ref{fig-empirical-cdf-infinite}.XXX
}
\label{fig-empirical-cdf}
\end{figure}

%XXX I think this plot is unnecesary
%as its doesn't provide any new insight beyond what was shown in
%figure \ref{fig-empirical-cdf}
\iffalse
\subsubsection{Comparison of risk bound and Monte Carlo risk}

In our next set of simulations we compare
the risk bound and the actual drawdown risk
$\prob(\Wmin < \alpha)$, computed by
Monte Carlo simulation, for
various choices of $\beta$. 
All other parameters are the same as before. 
%We also show the risk achieved by the QRCK bets, %(solutions of problem~\ref{e-qp-approx}), 
%for the same values of $\lambda$.
The results are shown in figure \ref{fig-risk-bound-montecarlo}.
We see that the risk bound is not bad, typically around
30\% or so higher than the actual risk. 
We have observed this to be the case across many problem instances.

\begin{figure}
\begin{center}
\includegraphics[width=.7\textwidth]{src/risk_bound_vs_MC_finite.png}
\end{center}
\caption{Comparison of drawdown risk bound $\beta$ and 
actual drawdown risk $\prob(\Wmin < 0.7)$ computed
by Monte Carlo simulation.
XXX Plot title should be changed to use $<$ not $\le $.XXX
}
\label{fig-risk-bound-montecarlo}
\end{figure}
\fi

\subsubsection{Comparison of RCK and QRCK}
Table \ref{tab-results-montecarlo-qrck} shows the Monte Carlo values of 
growth rate and drawdown risk for the QRCK bets with $\lambda = 6.456$
(to compare with the RCK solution in table \ref{tab-results-montecarlo}).
The QRCK bets come with no guarantee on the drawdown risk,
but with $\lambda = 6.456$ the drawdown probability (evaluated by
Monte Carlo) is less than $\beta=0.1$.
The value $\lambda = 2.800$ is selected so that 
the risk is approximately $0.10$; we see that
its growth rate is smaller than 
the growth rate of the RCK bet with the same drawdown risk.
\begin{table}
\begin{center}
\begin{tabular}{l|c|c|c}
Bet & $\Expect \log (r^Tb)$ & $e^{\lambda \log \alpha}$ &
%$\Expect (r^Tb)^{-\lambda}$ bound &
$\prob(\Wmin < \alpha)$\\ \hline
%Kelly  & 0.062  & -  & 0.397 \\
%RCK, $\lambda =6.456$  & 0.043  & 0.100  & 0.073 \\
%RCK, $\lambda =5.500$  & 0.047  & 0.141  & 0.099 \\
QRCK, $\lambda = 0.000$ & 0.054  & 1.000  & 0.218 \\ 
QRCK, $\lambda = 6.456$ & 0.027  & 0.100  & 0.025 \\
QRCK, $\lambda = 2.800$ & 0.044  & 0.368  & 0.100 \\
\end{tabular}
\caption{Statistics for QRCK bets. 
Expected growth rate
and drawdown risk are computed with Monte Carlo simulations.}
\label{tab-results-montecarlo-qrck}
\end{center}
\end{table} 

Figure \ref{fig-rcq-qrck-positions-comparison} shows the values of each $b_i$ with $i = 1, \ldots, 20$,
 for the Kelly, RCK, and QRCK bets. 
We can see that the Kelly bet concentrates on outcome~4; the RCK and QRCK 
bets still make a large bet on outcome~4, but also spread their bets 
across other outcomes as well.

\begin{figure}
\begin{center}
\includegraphics[width=1.\textwidth]{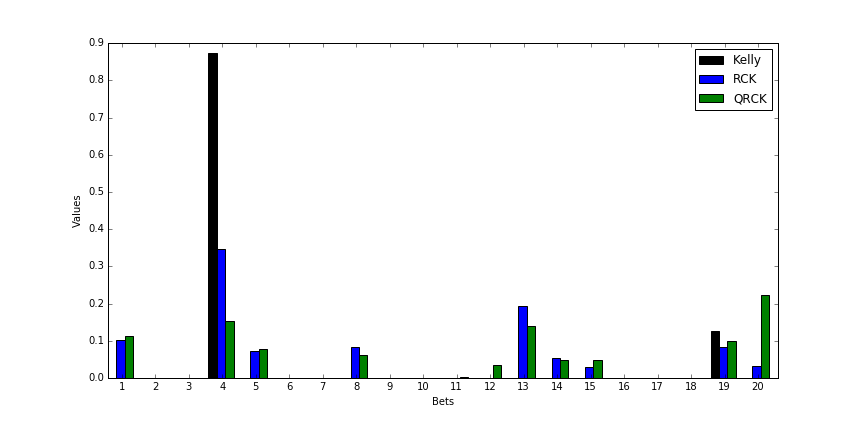}
\end{center}
\caption{Comparison of the values $b_i$, $i = 1, \ldots, 20$, for different bet vectors. 
The RCK and QRCK bets are both obtained with $\lambda = 6.456$.}
\label{fig-rcq-qrck-positions-comparison}
\end{figure}

\subsubsection{Risk-growth trade-off of RCK, QRCK, and fractional Kelly bets}
In figure \ref{fig-risk-restr-frac} we compare the 
trade-off between drawdown risk and expected growth rate
of the RCK and QRCK problems for multiple choices of $\lambda$
and the fractional Kelly bets \eqref{e-fractional-kelly} for multiple choices of $f$.
The plots are akin to the risk-return tradeoff curves that are
typical in Markowitz portfolio optimization. We see that RCK
 yields superior bets than QRCK and fractional Kelly 
(in some cases substantially better).  For example, the Kelly 
fractional bet that achieves our risk bound $0.1$ has a growth rate
around $0.035$, compared with RCK, which has a growth rate $0.047$.
\begin{figure}
\begin{center}
\includegraphics[width=.7\textwidth]{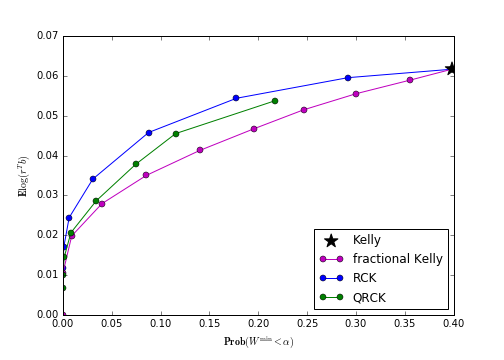}
\end{center}
\caption{Trade-off of drawdown risk and expected growth rate
for the RCK, QRCK, and fractional Kelly bets. The
Kelly bet is also shown.}
\label{fig-risk-restr-frac}
\end{figure}

% this clears the floats but doesn't insert a page break
\FloatBarrier

\subsection{General returns}
We show here an instance of the problem with an infinite distribution of returns,
defined as a mixture of lognormals
\[
r \sim 
\begin{cases}
\log \mathcal{N}(\nu_1, \Sigma_1) & \mbox{w.p. } 0.5\\ 
\log \mathcal{N}(\nu_2, \Sigma_2) & \mbox{w.p. } 0.5,
\end{cases}
\]
with $n = 20$, $\nu_1, \nu_2 \in \reals^n$, and $\Sigma_1, \Sigma_2 \in \symm_{+}^n$. 
We have $\nu_{1n} = \nu_{2n} = 0$, and
the matrices $\Sigma_1, \Sigma_2$ are such that $r_n$ has value 1 with probability 1.

We generate a sample of $10^6$ observations from this returns distribution and use it to solve the
Kelly, RCK, and QRCK problems. In the algorithms for solving Kelly and RCK we use
step sizes $t_k = C/ \sqrt{k}$ for some $C>0$ and batching over $100$ samples per iteration 
(so that we run $10^4$ iterations). 
We initialize these algorithms with the QRCK solutions to speed up convergence. 
%(as discused in \S\ref{subsec-infinite-convergence}).
To solve the QRCK problem we compute the first and second moments of $\rho = r-1$ 
using the same sample of $10^6$ observations.
We generate a separate sample of returns for Monte Carlo simulations, 
consisting of $10^4$ simulated trajectories of $w_t$ for $t = 1, \ldots, 100$.

\subsubsection{Comparison of RCK and Kelly bets}
In our first test we compare the Kelly bet with RCK bets for two values of the 
parameter $\lambda$. The first RCK bet has $\lambda = 6.456$, which guarantees
that the drawdown risk at $\alpha = 0.7$ is smaller or equal than $\beta = 0.1$. 
Table \ref{tab-results-montecarlo-infinite}
shows that this is indeed the case, the Monte Carlo simulated risk is $0.08$,
not very far from the theoretical bound. The second value of $\lambda$ is instead chosen
so that the Monte Carlo risk is approximately equal to $0.1$.
%All these values are consistent with what we observed in the finite outcome case. 
\begin{table}
\begin{center}
\begin{tabular}{l|c|c|c}
Bet & $\Expect \log (r^Tb)$ & $e^{\lambda \log \alpha}$ &
%$\Expect (r^Tb)^{-\lambda}$ bound &
$\prob(\Wmin < \alpha)$\\ \hline
Kelly  & 0.077  & -  & 0.569 \\
RCK, $\lambda = 6.456$  & 0.039  & 0.100  & 0.080 \\
RCK, $\lambda = 5.700$  & 0.043  & 0.131  & 0.101\\
\end{tabular}
\caption{Comparison of Kelly and RCK for the infinite outcome case.
Expected growth rate
and drawdown risk are computed with Monte Carlo simulations.}
\label{tab-results-montecarlo-infinite}
\end{center}
\end{table}

Figure \ref{fig-time-domain-trajectories-infinite} shows 10 simulated 
trajectories $w_t$ for the Kelly bet and for the RCK bet with $\lambda = 5.700$.
In this case 6 of the Kelly trajectories fall below $\alpha = 0.7$ and 
one of the RCK trajectories does, consistently with the values obtained above. 
\begin{figure}
\begin{center}
\includegraphics[width=1.\textwidth]{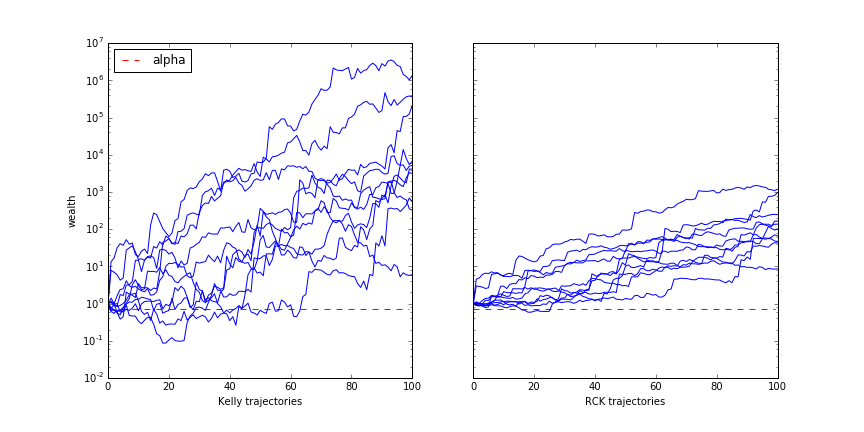}
\end{center}
\caption{Wealth trajectories for the Kelly optimal bet (left)
and the RCK bet with $\lambda=5.700$,
for the infinite returns distribution.
The wealth threshold $\alpha = 0.7$ is also shown.}
\label{fig-time-domain-trajectories-infinite}
\end{figure}
Figure \ref{fig-empirical-cdf-infinite} shows the sample CDF of $\Wmin$
for the Kelly bet and the RCK bet with $\lambda = 6.456$, and the theoretical bound given by $\alpha^\lambda$. 
%We observe that the bound holds across all values of $\Wmin$ and is reasonably tight.
\begin{figure}
\begin{center}
\includegraphics[width=.7\textwidth]{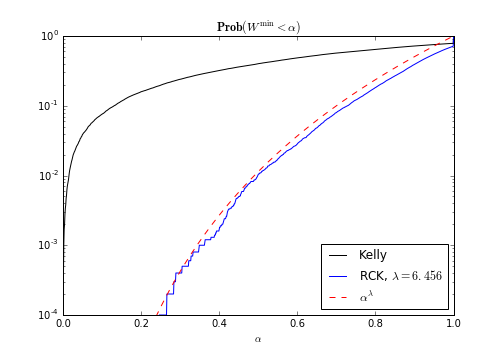}
\end{center}
\caption{Sample CDF of $\Wmin$ 
for the Kelly optimal bet and RCK bet
with $\lambda = 6.46$, for the infinite returns distribution.
The bound $\alpha^\lambda$ is also shown.}
\label{fig-empirical-cdf-infinite}
\end{figure}

\subsubsection{Risk-growth trade-off of RCK, QRCK, and fractional Kelly bets}
We compare the Monte Carlo simulated drawdown risk ($\alpha = 0.7$)
and expected growth rate
of the Kelly, RCK, QRCK, and fractional Kelly bets.
We select multiple values of $\lambda$ (for RCK and QRCK)
and $f$ (for fractional Kelly) and plot them together in 
figure \ref{fig-frontiers-infinite}. 
We observe, as we did in the finite outcome case, that RCK yields superior
bets than QRCK. This is particularly significant since QRCK is closely 
connected to the classic Markowitz portfolio optimization model, and this 
example resembles a financial portfolio selection problem (a bet with 
infinite return distributions). 
The fractional Kelly bets, in this case, show instead
the (essentially) same performance as RCK.
\begin{figure}
\begin{center}
\includegraphics[width=.7\textwidth]{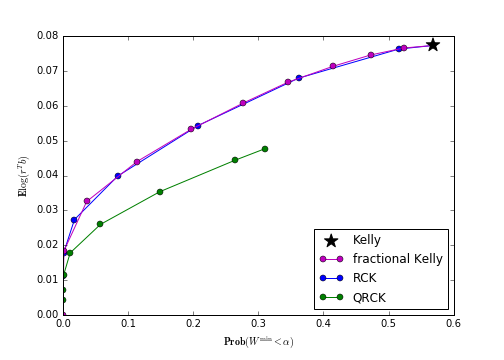}
\end{center}
\caption{Trade-off of drawdown risk and expected growth rate
for the RCK, QRCK, and fractional Kelly bets, with the infinite return
distribution. The Kelly bet is also shown.}
\label{fig-frontiers-infinite}
\end{figure} 

% The optimal value of the simple Kelly gambling problem \eqref{e-kelly} is 0.071. 
% We fix $\alpha = 0.75$ and $\beta  = 0.01$ which corresponds to a probability of 1\% (or lower)
% or losing 25\% of the initial wealth.
% The figure shows the optimal value of the risk-constrained problem for various choices of $\lambda$. 
% \begin{center}
% \includegraphics[width=.75\textwidth]{src/sweep_lambda.png}
% \end{center}

% We select the optimal $\lambda$ and generate 10000 wealth trajectories of lenght $T = 500$. For each
% one we select its minimum element. We thus build a sample distribution of the minimum element of the
% wealth trajectories. We plot these estimates of $\prob(\Wmin \leq \alpha)$ 
% for the unconstrained and constrained solutions.
% \begin{center}
% \includegraphics[width=.75\textwidth]{src/empirical_CDF.png}
% \end{center}
% We also plot the logarithm of the CDF (base 10).
% \begin{center}
% \includegraphics[width=.75\textwidth]{src/log10_empirical_CDF.png}
% \end{center}

\clearpage

\appendix

\section{Miscellaneous lemmas}
Here we collect the technical details and derivations of several results 
from the text.

%\begin{lem}
%Assume $\Expect r_i<\infty$ for $i=1,\dots,n-1$.
%Then $\Expect \log (r^Tb)<\infty$ for all $b\in \{b \mid \ones^Tb=1,~b \geq 0\}$.
%\end{lem}
%\begin{proof}
%By Jensen's inequality,
%\[
%\Expect \log(r^Tb)\le \log \Expect r^Tb< \infty.
%\]
%\end{proof}

\begin{lem}\label{lem-loser-bet}
The bet $e_n$ is 
optimal for problem~\eqref{e-kelly}
if and only if
$\Expect r_i \leq 1$ for $i=1, \ldots, n-1$.
Likewise, the bet $e_n$ is 
optimal for problem~\eqref{e-kelly-risk-restr}
if and only if
$\Expect r_i \leq 1$ for $i=1, \ldots, n-1$.
\end{lem}
\begin{proof}
Assume $\Expect r_i\le 1$ for $i=1,\dots,n-1$.
Then by Jensen's inequality,
\[
\Expect \log (r^Tb)\le \log \Expect r^Tb\le 0
\]
for any $b\in \{b \mid \ones^Tb=1,~b \geq 0\}$.
So $e_n$, which achieves objective value $0$, is optimal
for both problem~\eqref{e-kelly}
and problem~\eqref{e-kelly-risk-restr}.

Next assume $\Expect r_i>1$ for some $i$.
Consider the bet $b=fe_i+(1-f)e_n$  and
\[
\phi(f)=\Expect\log(r^Tb)
\]
for $f\in[0,1)$.
By \cite[Proposition~2.1]{bertsekas1973} we have the right-sided derivative
\[
\partial_+\phi(f)=\Expect\frac{r(e_i-e_n)}{r^Tb},
\]
and we have
\[
\partial_+\phi(0)=\Expect r(e_i-e_n)=\Expect r_i-1>0.
\]
Since $\phi(0)=0$, we have $\phi(f)>0$ for a small enough $f$.
So $b=fe_i+(1-f)e_n$ has better objective value than $e_n$ for small enough $f$,
\ie, $e_n$ is not optimal for problem~\eqref{e-kelly}.

Again, assume $\Expect r_i>1$ for some $i$.
Write $b^\star$ for the Kelly optimal bet. (We have already established that
$b^\star\ne e_n$ and that $\Expect \log (r^Tb^\star)>0$.)
Consider the bet $\tilde{b} = fb^\star + (1-f) e_n$ and
\[
\psi(f)= \Expect (r^T\tilde{b}) ^{-\lambda}
\]
for $f \in [0,1)$.
By a similar argument as before, we have
\[
\partial_+\psi(f) = -\lambda
\Expect (r^T\tilde{b})^{-\lambda -1} r^T(b^\star - e_n),
\]
and
\[
\partial_+\psi(0) = -\lambda \left( \Expect (r^T b^\star)  - 1\right).
\]
By Jensen's inequality we have
\[
0\le\Expect\log(r^Tb^\star)\le \log\Expect r^Tb^\star.
\]
So $\partial_+\psi(0)<0$.
Since $\psi(0)=1$, we have $\psi(f)<1$ for small enough $f$.
So $\tilde{b} = fb^\star + (1-f) e_n$ is feasible for small enough $f$.
Furthermore, $\tilde{b}$ has strictly better objective value than $e_n$ for $f\in (0,1)$
since the objective is concave.
So $e_n$ is not optimal.
\end{proof}

\begin{lem}\label{lem-dual-sol}
Problem \eqref{e-kelly-risk-restr} always has an optimal dual variable.
\end{lem}
\begin{proof}
Assume $e_n$ is optimal. By Lemma~\ref{lem-loser-bet}, 
$e_n$ is optimal even without the constraint $\Expect (r^Tb)^{-\lambda}\le 1$.
So $\kappa^\star=0$ is an optimal dual variable.

Now assume $e_n$ is not optimal.
By the reasoning with the function $\psi$ of Lemma~\ref{lem-loser-bet}'s proof,
there is a point strictly feasible with respect to the 
constraint $\Expect (r^Tb)^{-\lambda}\le 1$.
So by \cite[Theorem~17]{rockafellar1974}, we conclude a dual solution exists.
\end{proof}

\iffalse
\begin{lem}
The drawdown risk
%$\prob(W^\mathrm{min}\le \alpha)$
satisfies XXX
on $b\ge 0$.
\end{lem}
\begin{proof}
Since the geometric mean is a concave function,
\[
\{b\,|\,(\Pi^t_{i=1}r_i^Tb)^{1/t}>\alpha^{1/t}\}
=\{b\,|\,\Pi^t_{i=1}r_i^Tb>\alpha\}
\]
is a convex set for each $t=1,2,\dots$.
Since intersection of convex sets is convex, 
\[
S=\bigcap^\infty_{t=1}\{b\,|\,\Pi^t_{i=1}r_i^Tb>\alpha\}
=\{b\,|\,\Pi^t_{i=1}r_i^Tb>\alpha,\,t=1,2,\dots\}
\]
is convex.
Let $b_0=\theta b_1+(1-\theta)b_2$, where $\theta\in [0,1]$.
Since $S$ is convex, $b_1,\,b_2\in S$ implies $b\in S$.
Write
\[
W^\mathrm{min}(b)=\inf_{t=1,2,\dots}\Pi^t_{i=1}r_i^Tb.
\]
So if $W^\mathrm{min}(b_1)>\alpha$ and $W^\mathrm{min}(b_2)>\alpha$,
then $W^\mathrm{min}(b_0)>\alpha$.
We now have the union bound
\begin{align*}
\prob(W^\mathrm{min}(b_0)\le \alpha)
&\le 
\prob(W^\mathrm{min}(b_1)\le \alpha
\text{ or }
W^\mathrm{min}(b_2)\le \alpha)\\
&\le 
\prob(W^\mathrm{min}(b_1)\le \alpha)
+\prob(W^\mathrm{min}(b_2)\le \alpha).
\end{align*}

\end{proof}
\fi

\begin{lem}\label{lem-proj}
Let $\varepsilon\in[0,1]$,
and write $\Pi(\cdot)$ for the projection onto $\Delta_\varepsilon$.
Define the function $(\cdot)_{+,\varepsilon}:\reals^n\rightarrow\reals^n$ as
$((x)_{+,\varepsilon})_i=\max\{x_i,0\}$  for $i=1,\dots,n-1$
and
$((x)_{+,\varepsilon})_n=\max\{x_n,\varepsilon\}$.
Then
\[
\Pi(z)=(z-\nu\ones)_{+,\varepsilon},
\]
where $\nu$ is a solution of the equation
\[
\ones^T(z-\nu\ones)_{+,\varepsilon}=1.
\]
The left-hand side is a nonincreasing function of $\nu$,
so $\nu$ can be obtained efficiently via bisection on 
with the starting interval $[\max_iz_i-1,\max_iz_i]$.
\end{lem}
\begin{proof}
By definition of the projection, $x=\Pi(z)$ is the solution of
\[
\label{e-proj-prob}
\begin{array}{ll}
\mbox{minimize}&
\frac{1}{2}\|x-z\|_2^2\\
\mbox{subject to}
&
\ones^Tx=1\\
&
x\ge 0,\quad x_n\ge \varepsilon.
\end{array}
\]
This problem is equivalent to
\[
\begin{array}{ll}
\mbox{minimize}&
\frac{1}{2}\|x-z\|_2^2
+\nu(\ones^Tx-1)\\
\mbox{subject to}
&
x\ge 0,\quad x_n\ge \varepsilon
\end{array}
\]
for an optimal dual variable $\nu\in \reals$.
This follows from dualizing with respect to the constraint $\ones^Tx=1$ (and not the others),
applying strong Lagrange duality,
and using fact that the objective is strictly convex
\cite{bertsekas2009}.
This problem is in turn equivalent to
\[
\begin{array}{ll}
\mbox{minimize}&
\frac{1}{2}\|x-(z-\nu\ones)\|_2^2\\
\mbox{subject to}
&
x\ge 0,\quad x_n\ge \varepsilon,
\end{array}
\]
which has the analytic solution
\[
x^\star=(z-\nu\ones)_{+,\varepsilon}.
\]

An optimal dual variable $\nu$ must satisfy
\[
\ones^T
(z-\nu\ones)_{+,\varepsilon}=1,
\]
by the KKT conditions.
Write $h(\nu)=(z-\nu\ones)_{+,\varepsilon}$.
Then $h(\nu)$ a continuous nonincreasing function
with 
\[
h (\max_i z_i-1)\ge 1,\quad h(\max_iz_i)=\varepsilon.
\]
So a solution of $h(\nu)=1$ is in the interval
$[\max_iz_i-1,\max_iz_i]$.

\end{proof}

\begin{lem}\label{lem-opt-cond}
A pair $(b^\star,\kappa^\star)$ is a solution of the RCK problem
if and only if it satisfies conditions \eqref{e-opt-cond}.
\end{lem}
\begin{proof}
Assume
$(b^\star,\kappa^\star)$ is a solution of the RCK problem.
This immediately gives us $\ones^Tb^\star=1$, $b^\star\ge 0$,
$\Expect (r^Tb^\star)^{-\lambda}\le 1$, and $\kappa^\star\ge 0$.
Moreover $b^\star$ is a solution to
\BEQ\label{e-kelly-lagrange}
\begin{array}{ll}
\mbox{minimize}&
-\Expect \log (r^Tb)+\kappa^\star \Expect(r^Tb)^{-\lambda}\\
\mbox{subject to}& \ones^Tb = 1, \quad b \geq 0
\end{array}
\EEQ
and $\kappa^\star(\Expect (r^Tb^\star)^{-\lambda}-1)=0$,
by Lagrange duality.

Consider $b^\varepsilon=(1-\varepsilon)b^\star+\varepsilon b$,
where $b\in \reals^n$ satisfies $\ones^Tb=1$ and $b\ge 0$.
Define
$\varphi(\varepsilon)=
-\Expect \log (r^Tb^\varepsilon)+\kappa^\star \Expect(r^Tb^\varepsilon)^{-\lambda}$.
Since $b^\varepsilon$ is feasible for problem \eqref{e-kelly-lagrange} for $\varepsilon\in[0,1)$
and since $b^\star$ is optimal,
we have $\varphi(0)\le \varphi(\varepsilon)$ for $\varepsilon\in[0,1)$.
We later show $\varphi(\varepsilon)<\infty$
for $\varepsilon\in[0,1)$.
So $\partial_+\varphi(0)\ge 0$, where $\partial_+$ denotes the right-sided derivative.

By \cite[Proposition~2.1]{bertsekas1973} we have 
\[
\partial_+\varphi(\varepsilon)=
-\Expect \frac{r^T(b-b^\star)}{r^Tb^\varepsilon}-\kappa^\star\lambda \Expect\frac{r^T(b-b^\star)}{(r^Tb^\varepsilon)^{\lambda+1}}.
\]
So we have
\[
0\le
\partial_+\varphi(0)=
1
+
\kappa^\star\lambda \Expect (r^Tb^\star)^{-\lambda}
-\Expect \frac{r^Tb}{r^Tb^\star}-\kappa^\star\lambda \Expect\frac{r^Tb}{(r^Tb^\star)^{\lambda+1}}.
\]
Using $\kappa^\star(\Expect (r^Tb^\star)^{-\lambda}-1)=0$, and reorganizing we get
\BEQ\label{e-opt-cond2}
\left(
\Expect \frac{r}{r^Tb^\star}+\kappa^\star\lambda \Expect\frac{r}{(r^Tb^\star)^{\lambda+1}}
\right)^Tb
\le
1
+
\kappa^\star\lambda,
\EEQ
for any $b\in \reals^n$ that satisfies $\ones^Tb=1$ and $b\ge 0$.
With $b=e_i$ for $i=1,\dots,n$, we get
\BEQ\label{e-opt-cond3}
\Expect \frac{r_i}{r^Tb^\star}+\kappa^\star\lambda \Expect\frac{r_i}{(r^Tb^\star)^{\lambda+1}}
\le
1
+
\kappa^\star\lambda,
\EEQ
for $i=1,\dots,n$.

Now assume  $b^\star_i>0$ for some $i$ and let $b=e_i$.
Then $b^\varepsilon$ is feasible 
for problem \eqref{e-kelly-lagrange} for small negative $\varepsilon$.
We later show $\varphi(\varepsilon)<\infty$ for small negative $\varepsilon$.
This means $\partial_+\varphi(0)=0$ in this case, and we conclude
\[
\Expect \frac{r_i}{r^Tb^\star}+\kappa^\star\lambda \Expect\frac{r_i}{(r^Tb^\star)^{\lambda+1}}
~\left\{
\begin{array}{ll}
\le 1+\kappa^\star\lambda & b_i=0\\
= 1+\kappa^\star\lambda & b_i>0.
\end{array}\right.
\]

It remains to show that $\varphi(\varepsilon)<\infty$ for appropriate values of $\varepsilon$.
First note that $b^\star$, a solution, has finite objective value, \ie,
$-\Expect\log(r^Tb^\star)<\infty$ and $\Expect(r^Tb^\star)^{-\lambda}<\infty$.
So for $\varepsilon\in[0,1)$, we have
\begin{align*}
\varphi(\varepsilon)&=-\Expect \log((1-\varepsilon)r^Tb^\star+\varepsilon r^Tb)
+\kappa^\star
\Expect((1-\varepsilon)r^Tb^\star+\varepsilon r^Tb)^{-\lambda}\\
&\le
-\Expect \log((1-\varepsilon)r^Tb^\star)
+\kappa^\star
\Expect((1-\varepsilon)r^Tb^\star)^{-\lambda}\\
&=
-\log(1-\varepsilon)
-\Expect \log(r^Tb^\star)
+\kappa^\star(1-\varepsilon)^{-\lambda}
\Expect(r^Tb^\star)^{-\lambda}<\infty.
\end{align*}
Next assume $b_i>0$ for some $i\in \{1,\dots,n\}$. Then for $\varepsilon\in (-b_i,0)$
we have
\begin{align*}
\varphi(\varepsilon)
&\le
-\Expect \log(r^Tb^\star+\varepsilon r_i)
+\kappa^\star
\Expect(r^Tb^\star+\varepsilon r_i)^{-\lambda}\\
&\le
-\Expect \log(((b_i+\varepsilon)/b_i)r^Tb^\star)
+\kappa^\star
\Expect(((b_i+\varepsilon)/b_i)r^Tb^\star)^{-\lambda}\\
&=
-\log((b_i+\varepsilon)/b_i)
-\Expect \log(r^Tb^\star)
+\kappa^\star
((b_i+\varepsilon)/b_i)^{-\lambda}
\Expect(r^Tb^\star)^{-\lambda}<\infty.
\end{align*}

Finally, assume conditions \eqref{e-opt-cond} for $(b^\star,\kappa^\star)$,
and let us go through the argument in reverse order to show the converse.
Conditions  \eqref{e-opt-cond} implies
condition \eqref{e-opt-cond3}, which in turn implies
condition \eqref{e-opt-cond2} as $\ones^Tb=1$.
Note  $b^\star$ has finite objective value
becuase
$\Expect(r^Tb^\star)^{-\lambda}<\infty$ by assumption
and 
$-\Expect \log (r^Tb^\star)\le (1/\lambda)\log \Expect (r^Tb^\star)^{-\lambda}<\infty$
by Jensen's inequality.
So $\varphi(0)\le\varphi(\varepsilon)$ for $\varepsilon\in[0,1)$
by the same argument as before, 
\ie, $b^\star$ is optimal for problem \eqref{e-kelly-lagrange}.
This fact, together with conditions \eqref{e-opt-cond},
give us the KKT conditions of the RCK problem,
and we conclude 
$(b^\star,\kappa^\star)$ is a solution of the RCK problem
by Lagrange duality \cite{bertsekas2009}.
\end{proof}

\begin{lem}\label{lem-rw}
Consider an IID sequence $X_1,X_2,\dots$ from probability measure $\mu$,
its random walk $S_n=X_1+X_2+\dots +X_n$,
a stopping time $\tau$,
and
\[
\psi(\lambda)=
\log \Expect \exp( -\lambda X)
=\log \int \exp(-\lambda x)\;d\mu(x).
\]
Then we have
\[
\Expect\left[
\exp(-\lambda S_\tau-\tau\psi(\lambda))
\mid\tau<\infty\right]
\prob(\tau<\infty)\le 1.
\]
\end{lem}
This lemma is a modification of an identity from
\cite{wald1944},
\cite[\S XVIII.2]{feller1971},
and \cite[\S 9.4]{gallager2013Stochastic}.
\begin{proof}
Consider the tilted probability measure
\[
d\mu_{\lambda}(x)=\exp(-\lambda x-\psi(\lambda))d\mu(x)
\]
and write $\prob_{\mu_\lambda}$ for the probability
under the tilted measure $\mu_\lambda$.
%Let $X_1',X_2',\dots$ be an IID sequence from probability measure $\mu_\lambda$,
%and let $S_n'=X_1'+X_2'+\dots +X_n'$.
%Let $\tau' = \inf\{t\ge 1\,|\,S_t'\le \alpha\}$.
Then we have
\begin{align*}
\prob_{\mu_\lambda}(\tau=n)&=
\int I_{\{\tau=n\}}
d\mu_{\lambda}(x_1,x_2,\dots,x_n)\\
&=
\int I_{\{\tau=n\}}
\exp(-\lambda s_n-n\psi(\lambda))d\mu(x_1,x_2,\dots,x_n)\\
&=
\Expect\left[
\exp(-\lambda S_\tau-\tau\psi(\lambda))
I_{\{\tau=n\}}\right]
\end{align*}
By summing through $n=1,2,\ldots$ we get
\begin{align*}
\prob_{\mu_\lambda}(\tau<\infty)
&=
\Expect\left[
\exp(-\lambda S_\tau-\tau\psi(\lambda))
I_{\{\tau<\infty\}}\right]\\
&=
\Expect\left[
\exp(-\lambda S_\tau-\tau\psi(\lambda))
\mid\tau<\infty\right]
\prob(\tau<\infty).
\end{align*}
Since $\prob_{\mu_\lambda}(\tau<\infty)\le 1$, we have the desired result.
\end{proof}

\section{DCP specification}
\label{s-dcp-spec}

The finite outcome RCK problem
\eqref{e-problem-finite-outcome}
can be formulated and solved in CVXPY as 
\begin{verbatim}
b = Variable(n)
lambda_risk = Parameter(sign = 'positive')
growth = pi.T*log(r.T*b)
risk_constraint = (log_sum_exp (log(pi) - lambda_risk * log(r.T*b)) <= 0)
constraints = [ sum_entries(b) == 1, b >= 0, risk_constraint ]
risk_constr_kelly = Problem(Maximize(growth),constraints)
risk_constr_kelly.solve()
\end{verbatim}
Here \verb|r| is the matrix whose columns are the return 
vectors, and \verb|pi| is the vector of probabilities.
The second to last line forms the problem (object), and in the 
last line the problem is solved.  The optimal bet is written into
\verb|b.value|.
We note that if we set \verb|lambda_risk| equal to 0 this problem formulation 
is equivalent (computationally) to the Kelly problem.

\subsection*{Acknowledgments}
We thank Andrea Montanari, B. Ross Barmish, and Minggao Gu
for useful discussions.

\bibliography{kelly}
\end{document}